\documentclass{llncs}
\usepackage{amsmath}
\usepackage{amssymb}
\usepackage{amsfonts}

\newcommand{\mytimes}{\mathord{\times}}
\newcommand{\mone}{\text{-}1}

\newcommand{\tspace}{\text{\textvisiblespace}}

\title{The Parameterized~Complexity of {Domination-type} Problems  and
 Application~to~Linear~Codes}

\author{David Cattan\'eo\inst{1}\and Simon Perdrix \inst{2}}
\institute{LIG UMR 5217, University of Grenoble, France\\ \email{David.Cattaneo@imag.fr}  \and CNRS, LORIA UMR 7503, Inria Project Team CARTE, Nancy, France\\ \email{Simon.Perdrix@loria.fr}}
\date{}
\begin{document}

\maketitle

\begin{abstract} 
We study the parameterized complexity of domination-type problems.
$(\sigma,\rho)$-domination is a general and unifying framework introduced by Telle: given $\sigma, \rho \subseteq \mathbb N$, a set $D$ of vertices of a graph $G$ is $(\sigma,\rho)$-dominating if for any $v\in D$, $|N(v)\cap D|\in \sigma$ and for any $v\notin D, |N(v)\cap D|\in \rho$. 
Our main result is that for any $\sigma$ and $\rho$ recursive sets,   deciding whether there exists a $(\sigma,\rho)$-dominating set of size $k$, or of size at most $k$, are both in W$[2]$. This general statement is optimal in the sense that several particular instances of $(\sigma,\rho)$-domination are W$[2]$-complete (e.g. \textsc{Dominating Set}). We prove the W$[2]$-membership for the \emph{dual} parameterization too, i.e. deciding whether there exists a $(\sigma,\rho)$-{dominating} set of size $n-k$ (or at least $n-k$) is in W$[2]$, where $n$ is the order of the input graph. We extend this result to a class of domination-type problems which do not fall into the $(\sigma,\rho)$-domination framework, including \textsc{Connected Dominating Set}. We also consider problems of coding theory which are related to domination-type problems with parity constraints.  In particular, we prove that the problem of the minimal distance of a linear code over $\mathbb F_q$ is in W$[2]$ when $q$ is a power of prime, for both standard and dual parameterizations, and W$[1]$-hard for the dual parameterization. 

To prove the W$[2]$-membership of the domination-type problems we extend the Turing-way to parameterized complexity by introducing a new kind of non-deterministic Turing machine with the ability to perform `blind' transitions, i.e. transitions which do not depend on the content of the tapes. We prove that the corresponding problem \textsc{Short Blind Multi-Tape Non-Deterministic Turing Machine} is W$[2]$-complete.  We believe that this new machine 
can be used to prove W$[2]$-membership of other problems, not necessarily related to domination.

\end{abstract}

\section{Introduction}
%

\subsubsection{Domination-type problems.} 
Domination problems are central in graph theory. Telle \cite{complexDomination} introduced the notion of $(\sigma,\rho)$-domination as a unifying framework for many problems of domination: for any  two sets of integers $\sigma$ and $\rho$, a set $D$ of vertices of a graph $G$ is \emph{$(\sigma,\rho)$-dominating} if for any vertex  $v\in D$, $|N(v)\cap D|\in \sigma$ and for any vertex $v\notin D$, $|N(v)\cap D| \in \rho$. 
Among others, dominating sets, independent sets, and perfect codes are some particular instances of $(\sigma, \rho)$-domination. 
When $\sigma, \rho \in \{\text{ODD}, \text{EVEN}\}$ (where $\text{EVEN}:=\{2n, n\in \mathbb N\}$ and $\text {ODD}:= \mathbb N\setminus \text{EVEN}$), $(\sigma,\rho)$-domination is strongly related to problems in coding theory such as finding the minimal distance of a linear code \cite{mod2domination}. 
Despite its generality, the $(\sigma,\rho)$-domination framework does not capture all the variants of domination. For instance, connected dominating set (i.e. a dominating set which induces a connected subgraph) does not fall into the $(\sigma,\rho)$-domination framework.

\subsubsection{Parameterized complexity of domination-type problems.} Most of the  domination-type problems are NP-hard \cite{complexDomination}, though some of them  are fixed-parameter tractable. We assume the reader is familiar with parameterized complexity and the W-hierarchy, otherwise we refer to \cite{monograph,flum2006parameterized}. The parameterized complexity of domination-type problems has been intensively studied \cite{sigmarho,otherDom,regSubgraph} since the seminal paper by Downey and Fellows \cite{fptc1}. For instance, \textsc{Dominating Set} is known to be W$[2]$-complete \cite{fptc1}, whereas \textsc{Independent Set} and \textsc{Perfect Code} are W$[1]$-complete \cite{fptc1,perfectCode} (see Figure \ref{tabDomination} for a list of  domination-type problems with their parameterized complexity). Another example is \textsc{Total Dominating Set} which is known to be W$[2]$-hard \cite{fpintractDom}. Parameterized complexity of domination-type problems  with parity constraints -- and as a consequence the parameterized complexity of the corresponding problems in coding theory -- has been studied in \cite{paraCompCode}: \textsc{OddSet} and \textsc{Weight Distribution} are W$[1]$-hard and in W$[2]$, whereas \textsc{EvenSet} and \textsc{Minimal Distance} are in W$[2]$. Additionally to these particular cases of domination-type problems, general results reveal how the parameterized complexity of $(\sigma,\rho)$-domination depends on the choice of $\sigma$ and $\rho$. For instance, Golovach et al. \cite{sigmarho} proved that when $\sigma\subseteq \mathbb N$ and $\rho\subseteq \mathbb N^+$ are non-empty finite sets, the problem of deciding whether a graph has a $(\sigma,\rho)$-dominating set of size greater than a fixed-parameter $k$ is W$[1]$-complete.  

In parameterized complexity, the choice of the parameter is decisive. For all the problems mentioned above the standard parameterization is considered, i.e. the parameter is the size of the solution, i.e. the $(\sigma,\rho)$-dominating set. Domination-type problems have also been studied according to  the dual parameterization, i.e. the parameter is the size of the $(\sigma,\rho)$-\emph{dominated} set. With the dual parameterization, the problem associated with $(\sigma,\rho)$-domination is FPT when $\sigma$ and $\rho$ are either finite or cofinite \cite{sigmarho}. As a consequence, \textsc{Independent Set}, \textsc{Dominating Set} and \textsc{Perfect Code} are FPT for the dual parameterization.  
With parity constraints (i.e. $\sigma, \rho \in \{\text{ODD}, \text{EVEN}\}$), the problem associated with $(\sigma, \rho)$-domination has been proved to be W$[1]$-hard \cite{sigmarho} for the dual parameterization. Attention was also paid to the parameterized complexity of $(\sigma,\rho)$-domination when parameterized by the tree-width of the graph \cite{sigRhoTreeWidth,treeWidth}.

\subsubsection{Our results.}  The main result of the paper is that for any $\sigma$ and $\rho$ recursive sets,  $(\sigma,\rho)$-domination belongs to W$[2]$ for the \emph{standard} parameterization i.e. $(\sigma,\rho)$-dominating set of size $k$ (and at most $k$).

This  general statement is optimal in the sense that problems of $(\sigma,\rho)$-domi\-nation are known to be W$[2]$-hard for some particular instances of $\sigma$ and $\rho$ (e.g. \textsc{Dominating Set}). We also prove that  for any $\sigma$ and $\rho$ recursive sets, $(\sigma,\rho)$-domination belongs to W$[2]$ for the \emph{dual} parameterization i.e. $(\sigma,\rho)$-dominating set of size $n-k$ (and at least $n-k$).
For several particular instances of $\sigma$ and $\rho$, the W$[2]$-membership was unknown: the standard parameterization of \textsc{Total Dominating Set} was not known to belong to W$[2]$, and neither did the dual parameterization of $(\sigma,\rho)$-domination for $\sigma,\rho\in \{\text{ODD},\text{EVEN}\}$. 

Moreover, we prove that \textsc{Strong Stable Set} (known to be in W$[1]$ \cite{sigmarho}) is W$[1]$-complete for the standard parameterization.

We also consider more general problems that do not fall into the $(\sigma,\rho)$-domination framework. For any property $P$ and any set $\rho$ of integers, $D$ is a $(P,\rho)$-dominating set in a graph $G$ if $(i)$ the subgraph induced by $D$ satisfies the property $P $ and $(ii)$ for any vertex $v\notin D$, $|N(v)\cap D|\in \rho$. A connected dominating set corresponds to $\rho = \mathbb N^+$ and $P$ being the property that the graph is connected. We prove that the standard parameterization of $(P,\rho)$-domination is in W$[2]$ i.e. $(P,\rho)$-dominating set of size $k$ (and at most $k$) for any $P$ and $\rho$ recursive. As a consequence, \textsc{Connected Dominating Set} is W$[2]$-complete. We also prove that another domination problem, \textsc{Digraph Kernel}, is W$[2]$-complete.

Finally, regarding problems in linear coding theory, we show that  the dual parameterization of  \textsc{Weight Distribution} and \textsc{Minimal Distance} are both in W$[2]$. We also consider extensions of these two problems from the field $\mathbb F_2$ to $\mathbb F_q$ for any power of prime $q$, and show  that \textsc{Weight Distribution over $\mathbb F_q$} is W$[1]$-hard and in W$[2]$ for both standard and dual parameterizations; and that \textsc{Minimal Distance over $\mathbb F_q$} is in W$[2]$ for the standard parameterization, and W$[1]$-hard and in W$[2]$ for the dual parameterization. 

Our contributions are summarized in Figure \ref{tabDomination}.

\subsubsection{Our approach: extending the Turing way to parameterized complexity.} The Turing way to parameterized complexity \cite{turingWay} consists in solving a problem with a particular kind of Turing machine to prove that the problem belongs to some class of the W-hierarchy. For instance, if a problem can be solved by a single-tape non-deterministic Turing machine in  a number of steps which only depends on the parameter, then the problem is in W$[1]$. The W$[1]$-membership of  \textsc{Perfect Code} has been proved using such a Turing machine \cite{perfectCode}. When the problem is solved by a multi-tape non-deterministic machine  in  a number of steps which only depends on the parameter, it proves that the problem is in W$[2]$.  
To prove the W$[2]$-membership of $(\sigma,\rho)$-domination for any $\sigma$ and $\rho$ when parameterized by the size of the solution, we introduce an extension  of the multi-tape non-deterministic Turing machine by allowing `blind' transitions, i.e. transitions which do not depend on the symbols pointed out by the heads. We show that the extra capability of doing blind transitions does not change the computational power of the machine in terms of parameterized complexity by proving that the problem \textsc{Short Blind Multi-Tape Turing Machine} is W$[2]$-complete. Blindness of the transitions makes the design of the Turing machine far more easier; moreover it seems that there is no simple and efficient simulation of the blind transitions using the standard Turing machine, even though a (not necessarily simple) efficient simulation exists because of the W$[2]$-completeness of \textsc{Short  Multi-Tape Turing Machine}.   
For these reasons, we believe that the blind Turing machine can be used to prove W$[2]$-membership of other problems, not necessarily related to domination-type problems.

\subsubsection{The paper is organized as follows:} the next section is dedicated to the introduction of the blind multi-tape Turing machine and the proof that the corresponding parameterized problem is W$[2]$-complete. In Section \ref{sigmarho}, several results on the parameterized complexity of $(\sigma,\rho)$-domination are given. In Section \ref{otherdom}, the parameterized complexity of domination-type  problems which do not fall in the $(\sigma, \rho)$-domination framework are given. Finally,  Section \ref{codes} is dedicated to problems from coding theory which are related to domination-type problems with parity conditions.

\begin{figure*}
\begin{center}
\tabcolsep=2pt
\begin{tabular}{c|c|c|c}
\hline
\multicolumn{4}{c}{$(\sigma,\rho)$-Domination}\\
\hline
\hline
Name & $(\sigma,\rho)$ Formulation & Standard & Dual\\
\hline
\textsc{Dominating Set} & $(\mathbb{N},\mathbb{N}^+)$ & $\text{W}[2]$-complete \cite{fptc1} & FPT \cite{sigmarho}\\
\hline
\textsc{Independent Set} & $(\{0\},\mathbb{N})$ & $\text{W}[1]$-complete \cite{fptc2} & FPT \cite{sigmarho}\\
\hline
\textsc{Perfect Code} & $(\{0\},\{1\})$ & $\text{W}[1]$-complete \cite{fptc2,perfectCode} & FPT \cite{sigmarho}\\
\hline
\textsc{Strong Stable Set} & $(\{0\},\{0,1\})$ &\begin{tabular}{c}
\textbf{\text{W}[1]-complete}\\ ($\text{W}[1]$ \cite{sigmarho})
\end{tabular}
& FPT \cite{sigmarho}\\
\hline
\textsc{Total Dominating Set} & $(\mathbb{N}^+,\mathbb{N}^+)$ & \begin{tabular}{c}
\textbf{\text{W}[2]-complete}\\($\text{W}[2]$-hard \cite{fpintractDom}) \end{tabular} & FPT \cite{sigmarho}\\
\hline
\multicolumn{2}{c|}{\begin{tabular}{c} $(\sigma$,\textsc{\text{ODD}}$)$\textsc{-Dominating Set},\\
$\sigma\mathord \in \{$\textsc{\text{ODD},\text{EVEN}}$\}$ \end{tabular}} & $\text{W}[1]$-hard, $\text{W}[2]$\protect\footnotemark[1]  & \begin{tabular}{c}
$\text{W}[1]$-hard \cite{sigmarho},\\ \textbf{\text{W}[2]}
\end{tabular}\\
\hline
\multicolumn{2}{c|}{\begin{tabular}{c}$(\sigma,\text{EVEN})$\textsc{-Dominating Set},\\ 
$\sigma \mathord \in\{\text{ODD},\text{EVEN}\}$   
\\ \end{tabular}} & $\text{W}[2]$\protect\footnotemark[1] & \begin{tabular}{c}$\text{W}[1]$-hard \cite{sigmarho},\\ \textbf{\text{W}[2]}\end{tabular}\\
\hline
\multicolumn{2}{c|}{\begin{tabular}{c}
$(\sigma,\rho)$\textsc{-Dominating Set},\\ when $\sigma,\rho$ recursive
\end{tabular}} & \textbf{\text{W}[2]} & \textbf{\text{W}[2]}\\
\hline
\hline
\multicolumn{4}{c}{Other Domination Problems}\\
\hline
\hline
\multicolumn{2}{c|}{
\begin{tabular}{c}
\textsc{Connected Dominating Set}\\(Dual: \textsc{Maximal Leaf Spanning Tree})
\end{tabular}
} & \begin{tabular}{c}
\textbf{\text{W}[2]-complete}\\($\text{W}[2]$-hard \cite{connectedDomination})
\end{tabular} & FPT \cite{Galbiati199445}\\
\hline
\multicolumn{2}{c|}{\textsc{Digraph Kernel}} & \begin{tabular}{c}
\textbf{\text{W}[2]-complete}\\($\text{W}[2]$-hard \cite{kernelDigraph})
\end{tabular} & Unknown\\
\hline
\hline
\multicolumn{4}{c}{Problems in Coding Theory}\\
\hline
\hline
\multicolumn{2}{c|}{\textsc{Weight Distribution}} &  $\text{W}[1]$-hard,$\text{W}[2]$ \cite{paraCompCode} &
\begin{tabular}{c}
$\text{W}[1]$-hard \cite{sigmarho},\\ \textbf{\text{W}[2]}
\end{tabular}\\
\hline
\multicolumn{2}{c|}{\textsc{Minimum Distance}} & $\text{W}[2]$ \cite{paraCompCode} &
\begin{tabular}{c}
$\text{W}[1]$-hard \cite{sigmarho},\\ \textbf{\text{W}[2]}
\end{tabular}\\
\hline
\multicolumn{2}{c|}{ \begin{tabular}{c}\textsc{Weight Distribution Over} $\mathbb{F}_q$,\\ ($q$ power of prime)\\ \end{tabular}} &  \textbf{\text{W}[1]-hard,\text{W}[2]} & \begin{tabular}{c}\textbf{\text{W}[1]-hard},\\ \textbf{\text{W}[2]}\end{tabular}\\
\hline
\multicolumn{2}{c|}{\begin{tabular}{c}\textsc{Minimum Distance Over} $\mathbb{F}_q$,\\ ($q$ power of prime)\\ \end{tabular}}& \textbf{\text{W}[2]} & \begin{tabular}{c}\textbf{\text{W}[1]-hard},\\ \textbf{\text{W}[2]}\end{tabular}\\\hline
\end{tabular}
\end{center}
\caption{Overview of the parameterized complexity of domination-type problems and some problems from coding theory. The `Standard' column corresponds to a parameterization by the size of the $(\sigma, \rho)$-dominating set (or the Hamming weight for the problems in coding theory). In this column we consider the problem of $(\sigma,\rho)$-dominating set of size $k$ and at most $k$ except for \textsc{Independent Set} and \textsc{Strong Stable Set}  which are  considered for the equality case only. The `Dual' column corresponds to the dual parameterization, e.g. parameterized by the size of the $(\sigma, \rho)$-dominated set for domination-type problems. Our contributions, depicted in bold font, improve the results indicated in parenthesis. 
\label{tabDomination}}
\end{figure*}
\footnotetext[1]{When parameterized by the size of the dominating set, the parameterized complexity of $(\sigma,\text{ODD})$-domination (resp. $(\sigma, \text{EVEN})$-domination) for $\sigma\in \{\text{ODD},\text{EVEN}\}$ can be derived from the parameterized complexity of \textsc{OddSet} (resp. \textsc{EvenSet}) which has been proved in \cite{sigmarho}. }

\section{Blind Multi-Tape Non-Determinis\-tic Turing Machine}

A \emph{blind} Turing machine is a Turing Machine able to do `blind' transitions, i.e. transitions which  do not depend on  the symbol under the head.
Blind transitions are of interest in the multi-tape case when the size of the Turing machine (i.e. the number of defined transitions) matters, since a single blind transition can be seen as a shortcut for up to $|\Gamma|^m$ transitions, where $\Gamma$ is the alphabet and $m$ the number of tapes. For the description of the transitions of a blind m-tape Turing Machine $M=(Q,\Gamma,\Delta,\Sigma,b,q_I,Q_A)$, we introduce a neutral symbol `$\tspace$' and define the transitions as:
$\Delta \subseteq \underline \Gamma^m \times Q \times \underline \Gamma^m \times Q \times \{(-1),0,(+1)\}^m$, where $\underline \Gamma =  \Gamma \cup \{\tspace\}  $. A neutral symbol on the left part means that the transition can be applied whatever the symbol of the alphabet on the corresponding tape is, and a neutral symbol on the right part means that the symbol on the tape is kept. For instance $\langle \tspace\ \tspace,q,aa,q',00\rangle$ is a blind transition of a $2$-tape machine which, whatever the symbols under the heads are, changes the internal state $q$ into $q'$ and writes `$a$' on both tapes. $\langle \tspace^m,q,\tspace^m,q,1^m\rangle$ (where $\sigma^m$ stands for $\sigma,\ldots,\sigma$, $m$ times) is a blind transition of a $m$-tape machine which moves all the $m$ heads to the right without modifying the content of the tapes. 

The parameterized problem associated with the Blind Multi-Tape Non-Deter\-mi\-nis\-tic Turing Machines is defined as:

\vspace{0.3cm}
\noindent \textsc{Short Blind Multi-Tape Non-Deter\-ministic Turing Machine Computation}\\
Input: A blind $m$-tape non-deterministic Turing Machine $M$, a word $w$ on the alphabet $\Sigma$, an integer $k$.\\
Parameter: $k$.\\
{Question: Is there a computation of $M$ on $w$ that reaches an accepting state in at most $k$ steps?}
\vspace{0.3cm}

\begin{theorem}
\label{theorem:SBMTNDTMW2}
\textsc{Short Blind Multi-Tape Non-Determi\-nistic Turing Machine Computation} is complete for $\text{W}[2]$.
\end{theorem}

\begin{proof}
The hardness for $\text{W}[2]$ comes from the  non-blind case which has been proven to be complete for $\text{W}[2]$ \cite{mTNDTM}. The proof of the $\text{W}[2]$-membership is similar to the non-blind case \cite{turingWay} and consists in a reduction to \textsc{Weighted Weft-2 Circuit Satisfiability}. This  problem consists in deciding whether a weft-2 mixed-type boolean circuit of depth bounded by a function of the parameter $k$, accepts some input of Hamming weight $k$.
A mixed type circuit is composed of `small' gates of fan-in $\le$ 2 and `large' AND and OR gates of unbounded fan-in.  
The weft of the circuit is the maximum number of unbounded fan-in gates on an input/output path.

First, we transform $M$ into a machine which accepts its input in (exactly) $k$ steps iff $M$ accepts its input in at most $k$ steps. To this end, all accepting states of $M$ are merged into a fresh accepting state $q_{A}$ and the blind transition $\langle \tspace^m,q_{A},\tspace^m,q_{A}, 0^m\rangle$ is added.

In the following, a weft-$2$ mixed circuit $C$ is constructed in such a way that the accepted inputs correspond to the sequences of $k$ transitions of a machine $M$ from the initial state to the accepting state. The set $\Delta$ of the transitions of $M$ are indexed by $j\in [1,|\Delta|]$.
The symbols of $\underline{\Gamma}$ are indexed by $s\in [0,|{\Gamma}|]$, where $0$ is the index of the blank symbol and $|{\Gamma}|$ is the index of the neutral symbol `$\tspace$'. Let $x[i,j]$ for $i\in [1,k], j\in [1, |\Delta|]$ and $x[-1,-1]$ be the input wires of the circuit. For $i\in [1,k], j\in [1, |\Delta|]$, $x[i,j]$ is true if and only if the $i^{th}$ transition of the sequence is the transition indexed by $j$, $x[-1,-1]$ represents the constant $0$.  
The following gates encode some information about the transitions of $M$: $\forall i\in [1,k], \forall q\in [1, |Q|], \forall s\in [0, |\Gamma|], \forall t\in [1, m], \forall d\in \{\mone,0, 1\}$,
\begin{itemize}
\item $\tau_o(i,q)$ outputs true iff the initial state on the $i^{th}$ transition is $q$: $$\tau_o(i,q):=\bigvee_{j\in J_{q}} x[i,j]$$ where $J_{q}=\Delta \cap ({\underline \Gamma^m} \mytimes \{q\}\mytimes {\underline \Gamma^m} \mytimes Q \mytimes \{\mone,0, 1\}^m)$

\item $\tau_n(i,q)$ outputs true iff the final state on the $i^{th}$ transition is $q$:
$$\tau_n(i,q):=\bigvee_{j\in J'_{q}} x[i,j]$$ where $J'_{q}=\Delta \cap ({\underline \Gamma^m} \mytimes Q\mytimes {\underline \Gamma^m} \mytimes \{q\} \mytimes \{\mone,0, 1\}^m)$

\item $\sigma_o(i,s,t)$ outputs true iff either the symbol read by the $i^{th}$ transition on  tape $t$ is $s$, or the transition does not read the symbol on tape $t$ in the `blind' case $s=|{\Gamma}|$: $$\sigma_o(i,s,t):=\bigvee_{j\in J_{s,t}} x[i,j]$$ where $J_{s,t}=\Delta \cap ({\underline \Gamma} ^{t-1}\mytimes \{s\}\mytimes {\underline \Gamma} ^{m-t}\mytimes Q\mytimes {\underline \Gamma^m} \mytimes Q \mytimes \{\mone,0, 1\}^m)$

\item $\sigma_n(i,s,t)$ outputs true iff either the symbol written by the $i^{th}$ transition on  tape $t$ is $s$, or the transition does not write any symbol on tape $t$ in the `blind' case $s=|{\Gamma}|$: $$\sigma_n(i,s,t):=\bigvee_{j\in J'_{s,t}} x[i,j]$$ where $J'_{s,t}=\Delta \cap ({\underline \Gamma^m} \mytimes Q\mytimes {\underline \Gamma}^{t-1}\mytimes\{s\}\mytimes {\underline \Gamma} ^{m-t}\mytimes Q \mytimes \{\mone,0, 1\}^m)$

\item  $\mu(i,d,t)$ outputs true iff the head of $t$ has a movement $d$ on the $i^{th}$ transition: $$\mu(i,d,t):=\bigvee_{j\in J_{d,t}} x[i,j]$$ where 
$J_{d,t}  =  \Delta \cap ({\underline \Gamma^m} \mytimes Q\mytimes {\underline \Gamma^m} \mytimes Q \mytimes \{\mone,0, 1\}^{t-1}\mytimes \{d\}\mytimes \{\mone,0, 1\}^{m-t})$
\end{itemize}
Notice that most of these gates require unbounded fan-in OR gates in general.

The following gates encode the position of the heads and all the symbols in every cell of the tapes. These gates guarantee the correctness of the transition sequence. $\forall i\in [1,k],\forall  l \in [\text{-}k,k],\forall  t \in [1,m], \forall s \in [0,|\Gamma|-1]$,
\begin{itemize}
\item $\beta(i,l,t)$ outputs true iff the head of  tape $t$ is at position $l$ before step $i$. Since the transition sequence is of length $k$,  $l$ is  in the interval $[\text{-}k,k]$. The gate is defined as:
\begin{eqnarray*}
\nonumber
\beta(0,l,t)&:=&
\begin{cases}
1 & \text{if } l=0\\
0 & \text{otherwise}\\
\end{cases}\\
&&\\
\beta(i,l,t)&:=&(\beta(i\mathord -1,l,t) \wedge \mu(i\mathord -1,0,t))\\
&&\vee ~( \beta(i\mathord -1,l\mathord -1,t) \wedge \mu(i\mathord-1,1,t))\\
&&\vee ~(\beta(i\mathord -1,l\mathord+1,t) \wedge \mu(i\mathord-1,\text-1,t))
\end{eqnarray*}

\item $\sigma(i,l,s,t)$
outputs true iff the cell $l$ of tape $t$ contains the symbol $s$ before step $i$.  Let $w$ be the input word of the machine, located on tape 1.
\begin{eqnarray*}
\nonumber
\sigma(0,l,s,t)&:=&
\begin{cases}
1 & \text{if } \big((s \text{ is the index of } w[l]) \wedge (t=1) \wedge (0 \leq l <|w|)\big)\\
1 & \text{if } \big((s=0) \wedge (t\neq 1 \vee l<0 \vee l \ge |w|)\big)\\
0 & \text{otherwise}\\
\end{cases}\\&&\\
\sigma(i,l,s,t)&:=&(\neg\beta(i-1,l,t) \wedge \sigma(i-1,l,s,t))\\&&\vee~ (\beta(i-1,l,t) \wedge \sigma_n(i-1,s,t))\\
&&\vee~ (\beta(i-1,l,t) \wedge \sigma_n(i-1,|\Gamma|,t) \wedge \sigma(i-1,l,s,t))
\end{eqnarray*}
One can see in the definition of $\sigma(i,l,s,t)$ for $i>0$ that there are three different cases: either the head was not pointing at the cell $l$, so the symbol remains unchanged; or the head was pointing at the cell $l$, and the symbol has been written in the previous step; or the head was on the cell but the transition was blind, so the symbol was already $s$. 
\end{itemize}

Notice that these gates have a bounded fan-in, and that the recursion is on the number of transitions, so  their depth is bounded by the parameter $k$. Notice also that there is a polynomial number of such gates since there are $k\cdot 2k\cdot m$, $\beta$ gates and $k\cdot 2k\cdot |\underline{\Gamma}|\cdot m$, $\sigma$ gates.

All the information about the computation path has been encoded so the remaining gates check the validity of this transition sequence:
\begin{itemize}
\item $E:=E_0 \wedge E_1 \wedge E_2 \wedge E_3 \wedge E_4$ is the final  gate of the circuit. As a consequence, for any input accepted by the circuit, the following  conditions $E_0,\ldots, E_4$ must be  satisfied. 
\item $E_0:=\neg x[\mone,\mone]$ ensures that $x[\mone,\mone]$ is the constant $0$, so $\neg x[\mone,\mone]$ is the constant $1$ used by the other gates.
\item $E_1$ ensures that for every $i$, at most one wire among the block $x[i,1], \dots,$ $x[i,|\Delta|]$ is true, which means that at each step at most one transition is performed. $E_1$ is defined as:
$$E_1:=\bigwedge_{i=1}^{k} \bigwedge_{j=1}^{|\Delta|} \bigwedge_{j'=1,j'\neq j}^{|\Delta|} (\neg x[i,j]\vee\neg x[i,j'])$$
\item $E_2$ ensures that the initial state of each step is equal to the final state of the previous step. $E_2$ is defined as:
$$E_2:=\bigwedge_{i=2}^{k} \bigwedge_{q=1}^{|Q|} (\neg \tau_n(i-1,q) \vee \tau_o(i,q))$$
Notice that this formula encodes: $\forall i\in [2,k],\forall q{\in} [1,|Q|],\tau_n(i-1,q) \Rightarrow \tau_o(i,q)$.
\item $E_3$ ensures that the symbol read by a transition on a tape is either the one pointed out by the head or any symbol when the transition is blind. 
$$E_3:=\bigwedge_{i=1}^{k} \bigwedge_{t=1}^{m} \bigwedge_{l=-k}^{k} \bigwedge_{s=0}^{|\Gamma|-1} (\neg\beta(i,l,t) \vee \neg\sigma(i,l,s,t) \vee \sigma_o(i,s,t) \vee \sigma_o(i,|\Gamma|,t))$$
Notice that this formula encodes: $\forall i\in[1,k],\forall l\in [-k,k],\forall s\in [0,|\Gamma|],\forall t\in [1,m], (\beta(i,l,t) \wedge \sigma(i,l,s,t)) \Rightarrow (\sigma_o(i,s,t)\vee\sigma_o(i,|\Gamma|,t))$.
\item $E_4$ ensures that the initial state on the first step is $q_0$, the initial state of $M$ of index $0$, and that the last state is the accepting state $q_A$ of index $|Q|-1$. So $E_4$ is defined as:$$E_4:=\tau_o(0,0) \wedge \tau_n(k-1,|Q|-1)$$
\end{itemize}

All the $E_i, i\in [0,4]$ gates are independent, so every input-output path goes through at most one of these unbounded fan-in gates. Since it is also the case for the gates encoding the transitions, and that the $\sigma$ and $\beta$ gates are bounded fan-in gates, the weft of this circuit is 2. Since the only recursive gates have a depth bounded by the parameter, the depth of this circuit is bounded by the parameter. Notice also that the number of gates is polynomial in $|M|$. This circuit outputs true if and only if $M$ has an accepting computation path of length $k$ on the word $w$, i.e. if and only if $M$ has an accepting computation path of length at most $k$ on $w$. Therefore, \textsc{Short Blind Multi-Tape Non-Deterministic Turing Machine Computation} belongs to $\text{W}[2]$.\hfill$\Box$
\end{proof}

What is interesting is that although the blindness of the transition does not change the computational power of short multi-tape non-deterministic Turing machines, there is no simple way to simulate the blind machine with the original one. Indeed, intuitively, a blind transition on $m$ tapes is a short-cut for up to $|\Gamma|^m$ transitions, so a machine with no blind transition may have an exponentially larger size. A tape-by-tape (sequentialization) simulation would avoid this blow up of the number of transitions, but will not reach an accepting state within a number of steps depending only on the parameter.

\section{Parameterized complexity of $(\sigma,\rho)$-Domination}
\label{sigmarho}
In this section, we prove the central result of the paper: for any recursive sets $\sigma$ and $\rho$,  $(\sigma,\rho)$-domination belongs to W$[2]$ for both \emph{standard} and \emph{dual} parameterizations, i.e. the four problems which consists in deciding whether a graph has a $(\sigma,\rho)$-dominating set of size $k$; of size $n-k$; of size at most $k$; and of size at least $n-k$ are in W$[2]$ with respect to $k$. To this end, we show that for any $\sigma$, $\rho$ recursive sets,  these problems of $(\sigma, \rho)$-domination can be decided using a blind multi-tape Turing machine. The only assumption on $\sigma$ and $\rho$ is that they are  recursive, i.e. there exists a Turing machine which decides whether a given integer $j$ belongs to $\sigma$ (resp. $\rho$). 
\vspace{0.15cm}\\
\noindent$(\sigma,\rho)$-\textsc{Dominating Set of Size at Most} $k$:\\
Input: A graph $G=(V, E)$, an integer $k$.\\
Parameter: $k$.\\
Question: Is there a $(\sigma,\rho)$-dominating set $D\subseteq V$ such that $|D|\leq k$?
\vspace{0.15cm}

$(\sigma,\rho)$-\textsc{Dominating Set of Size} $k$, $(\sigma,\rho)$-\textsc{Dominating Set of Size} $n{-}k$, and $(\sigma,\rho)$-\textsc{Dominating Set of Size at Least} $n-k$ are defined likewise.

\begin{theorem}\label{sigmarho:w2}
For any recursive sets of integers $\sigma$ and $\rho$, $(\sigma,\rho)$-\textsc{Dominating Set of Size at Most} $k$ and $(\sigma,\rho)$-\textsc{Dominating Set of Size} $k$ belong to $\text{\emph{W}}[2]$.
\end{theorem}
\begin{proof} We prove that $(\sigma,\rho)$-\textsc{Dominating Set of Size at Most} $k$ is in W$[2]$, the proof that $(\sigma,\rho)$-\textsc{Dominating Set of Size} $k$ belongs to W$[2]$ is similar. 
Given two recursive sets $\sigma, \rho\subseteq \mathbb N$,  an integer $k$, and a graph $G=(\{v_1,\ldots, v_n\},E)$, we consider the following  $(n+1)$-tape Turing machine $M$, which decides whether $G$ has a $(\sigma,\rho)$-dominating set of size at most $k$. $M$ works in 3 phases (see an example in Figure \ref{figexample}): $(1)$ a subset $D$ of size at most $k$ is non-deterministically chosen and written on the first tape. Moreover, the first $k+1$ cells of the following $n$  tapes -- one tape for each vertex of the graph -- are filled with $0$s and $1$s  such that the $i^{th}$ cell of each tape is $1$ iff $i\in \rho$; $(2)$ The content of the tapes associated with the vertices in $D$ is  removed and replaced by the characteristic vector of $\sigma$, i.e. the $i^{th}$ cell is $1$ iff $i\in \sigma$. At the end of this second phase, all heads are located on the leftmost non-blank symbol; $(3)$ For each vertex $v$ in $D$, the heads of all the tapes associated with a neighbor of $v$ move to the right. At the end of this third phase, for every $v\in D$ (resp. $v\in \overline D$),  the head of the tape associated with $v$ reads $1$ iff $|N(v)\cap D| \in \sigma$ (resp. $|N(v)\cap D| \in \rho$), so $D$ is a $(\sigma,\rho)$-dominating set iff all heads but the first one read a symbol $1$. 

The actual description of the blind $(n+1)$-tape non-determi\-nistic Turing machine  is as follows: $M=(Q,\Gamma,\Delta,\Sigma,b,q_I,Q_A)$, where $\Gamma=\{\Box,0,1, v_1,\ldots ,v_n\}$, $b=\Box$, $\Sigma= \varnothing$, $Q=\{q_{r,s}~|~r\in[1,n+1], s\in[0,k]\} \cup\{q^{\textup{ret}}_{s}~|~s\in[1,k+1]\} \cup\{q^{\textup{sig}}_{i,s}~|~i\in[1,n],s\in[0,k]\} \cup\{q^{\textup{sig}},q^{\textup{end}}_\rho,q^{\textup{end}}_\sigma,q^{\textup{read}},q_A\}$,  $q_I=q_{1,0}$ and $Q_A=\{q_A\}$. Given an integer set $A$, $\overline{A}$ is the complementary set of $A$, it is defined as the only set such that $A \cap \overline{A}=\varnothing$ and $A \cup \overline{A}=\mathbb{N}$.
 The initial word $w$ is the empty word, so every cell initially contains the blank symbol $\Box$. The transitions are:\\

\noindent \textit{Phase 1 -- Initialization of D and $\rho$:}

\vspace{0.2cm}
\begin{tabular}{l cl}
$\langle \Box\Box^n,q_{r,s},v_i1^n,q_{i+1,s+1},(+1)(+1)^n\rangle$ &~~& 
$r\in[1,n], s\in\rho\cap[0,k-1], i\in[r,n]$\\
$\langle \Box\Box^n,q_{r,s},v_i0^n,q_{i+1,s+1},(+1)(+1)^n\rangle$ &&
$r\in[1,n], s\in\overline{\rho}\cap[0,k-1], i\in[r,n]$\\
$\langle \Box\Box^n,q_{r,s},\Box1^n,q_{r,s+1},0(+1)^n\rangle$ &&
$r\in[1,n+1], s\in\rho\cap[0,k-1]$\\
$\langle \Box\Box^n,q_{r,s},\Box0^n,q_{r,s+1},0(+1)^n\rangle$ &&
$r\in[1,n+1], s\in\overline{\rho}\cap[0,k-1]$\\
$\langle \Box\Box^n,q_{r,k},\Box1^n,q^{\textup{end}}_\rho,(-1)(-1)^n\rangle$ &&
$r\in[1,n+1]$, if $k\in\rho$\\
$\langle \Box\Box^n,q_{r,k},\Box0^n,q^{\textup{end}}_\rho,(-1)(-1)^n\rangle$ &&
$r\in[1,n+1]$, if $k\in\overline{\rho}$\\
$\langle v_i\tspace^n,q^{\textup{end}}_\rho,v_i\tspace^n,q^{\textup{end}}_\rho,(-1)(-1)^n\rangle$ &&
$i\in[1,n]$\\
$\langle \Box 1^n,q^{\textup{end}}_\rho,\Box 1^n,q^{\textup{end}}_\rho,0(-1)^n\rangle$\\
$\langle \Box 0^n,q^{\textup{end}}_\rho,\Box 0^n,q^{\textup{end}}_\rho,0(-1)^n\rangle$\\
$\langle \Box\Box^n,q^{\textup{end}}_\rho,\Box\Box^n,q^{\textup{sig}},(+1)(+1)^n\rangle$
\end{tabular}\\

\vspace{0.2cm}
The state $q_{r,s}$ means that $s-1$ vertices  among $v_1, \ldots, v_{r-1}$ have already been written on the first tape, $q^{\textup{end}}_\rho$ that the initializations of $D$ and $\rho$ are done and that the heads are going back to the leftmost non blank cell on every tape.\\

\noindent \textit{Phase 2 -- Initialization of $\sigma$:}

\vspace{0.2cm}

{
\begin{tabular}{l l}
\noindent $\langle v_i\tspace^n,q^{\textup{sig}},v_i\tspace^n,q^{\textup{sig}}_{i,0},00^n\rangle$ &$~ \quad i\in[1,n]$\\
$\langle v_i\tspace^n,q^{\textup{sig}}_{i,s},v_i\tspace^{i-1}1\tspace^{n-i},q^{\textup{sig}}_{i,s+1},0(+1)^n\rangle$ & $~\quad i\in[1,n], s\in\sigma\cap[0,k-1]$\\
$\langle v_i\tspace^n,q^{\textup{sig}}_{i,s},v_i\tspace^{i-1}0\tspace^{n-i},q^{\textup{sig}}_{i,s+1},0(+1)^n\rangle$ & $~\quad i\in[1,n], s\in\overline{\sigma}\cap[0,k-1]$\\
$\langle v_i\tspace^n,q^{\textup{sig}}_{i,k},v_i\tspace^{i-1}1\tspace^{n-i},q^{\textup{ret}}_{1},(+1)0^n\rangle$ & $~\quad i\in[1,n]$, if $k\in \sigma$\\
\end{tabular}

\begin{tabular}{l l}
$\langle v_i\tspace^n,q^{\textup{sig}}_{i,k},v_i\tspace^{i-1}0\tspace^{n-i},q^{\textup{ret}}_{1},(+1)0^n\rangle$ & $~\quad i\in[1,n]$, if $k\in\overline{\sigma}$\\
$\langle v_i\tspace^n,q^{\textup{ret}}_{s},v_i\tspace^{n},q^{\textup{ret}}_{s+1},0(-1)^n\rangle$ & $~\quad i\in[1,n],s\in[1,k]$\\
$\langle v_i\tspace^n,q^{\textup{ret}}_{k+1},v_i\tspace^n,q^{\textup{sig}}_{i,0},00^n\rangle$ & $~\quad i\in[1,n]$\\
$\langle \Box\tspace^n,q^{\textup{ret}}_{1},\Box\tspace^n,q^{\textup{end}}_\sigma,(-1)0^n\rangle$\\
$\langle v_i\tspace^n,q^{\textup{end}}_\sigma,v_i\tspace^n,q^{\textup{end}}_\sigma,(-1)0^n\rangle$ & $~\quad i\in[1,n]$\\
$\langle \Box\tspace^n,q^{\textup{end}}_\sigma,\Box\tspace^n,q^{\textup{read}},(+1)0^n\rangle$
\end{tabular}\\
}

\vspace{0.2cm}

The state $q^{\textup{sig}}_{s,i}$ means that the first $s$ symbols of the characteristic vector of $\sigma$ have been written on the tape associated with the vertex $v_i$.\\

\textit{Phase 3: Neighborhood Checking}

\vspace{0.2cm}

\noindent $\langle v_i\tspace^n,q^{\textup{read}},v_i\tspace^n,q^{\textup{read}},(+1)d_1\dots d_{n}\rangle$ 
$~\quad i\in[1,n]$, where $d_t{=}\begin{cases}{+}1 &\text {if $v_t{\in} N(v_i)$} \\ 0&\text{otherwise}\end{cases}$\\
$\langle \Box 1^n,q^{\textup{read}},\Box 1^n,q_A,00^n\rangle$\\

Since $\sigma$ and $\rho$ are recursive, their characteristic vector of length $k$ can be computed and written on the tapes in time $f(k)$ for some fixed function $f$. In the first phase $D$ and $\rho$ of size $k$ are written on the tapes and then the heads comes back so there are $2(k+1)$ steps. In the second phase $\sigma$ of size $k$ is written sequentially on at most $k$ tapes corresponding to the elements of $D$ and then the heads come back so there are at most $k(2k)$ steps. Finally the third phase goes through D and moves the heads on the tapes of the neighbours so there are at most $k$ steps. The number of transitions is polynomial in $|G|$ and the acceptance is made in at most $2(k+1)+k(2k+2)$ steps if a $(\sigma,\rho)$-dominating set of size at most $k$ exists. As a consequence,  $(\sigma,\rho)$-\textsc{Dominating Set of Size at Most} $k$ belongs to $\text{W}[2]$. Notice that the use of blind transitions in the third phase is crucial. Indeed, a naive simulation of any of these blind transitions uses $2^n$ non-blind transitions since the transition should be applicable for any of the $2^n$ possible configurations read by the heads of the machine. \hfill $\Box$
\end{proof}

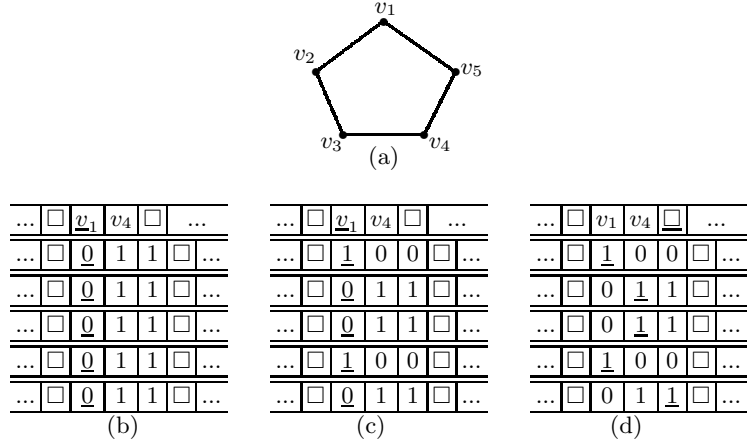
\begin{figure*}
\begin{center}
\tabcolsep=2pt
\begin{tabular}{c}

$\vcenter{
\ifx\JPicScale\undefined\def\JPicScale{0.6}\fi
\unitlength \JPicScale mm
\begin{picture}(40,33)(0,0)
\linethickness{0.3mm}
\put(6,-1){\text{$v_3$}}
\linethickness{0.3mm}
\put(0,18.5){\text{$v_2$}}
\linethickness{0.3mm}
\put(37,15.5){\text{$v_5$}}
\linethickness{0.3mm}
\put(18,29){\text{$v_1$}}
\linethickness{0.3mm}
\put(30,-1){\text{$v_4$}}
\linethickness{0.3mm}
\multiput(5,16)(0.16,0.12){92}{\line(1,0){0.16}}
\linethickness{0.3mm}
\multiput(20,27)(0.17,-0.12){92}{\line(1,0){0.17}}
\linethickness{0.3mm}
\multiput(29,2)(0.12,0.24){58}{\line(0,1){0.24}}
\linethickness{0.3mm}
\put(11,2){\line(1,0){18}}
\linethickness{0.3mm}
\multiput(5,16)(0.12,-0.28){50}{\line(0,-1){0.28}}

{
\put(29,2){\circle*{2}}

\linethickness{0.3mm}
\put(11,2){\circle*{2}}


\linethickness{0.3mm}
\put(5,16){\circle*{2}}

\linethickness{0.3mm}
\put(36,16){\circle*{2}}

\linethickness{0.3mm}
\put(20,27){\circle*{2}}
}
\end{picture}}$\\
(a)\\

\end{tabular}
\begin{tabular}{c c c}

\begin{tabular}{c|c|c|c|c|c|c}
\multicolumn{7}{c}{~}\\
\hline

... & $\Box$ & $\underline{v}_1$ & $v_4$ & $\Box$ & \multicolumn{2}{c}{...}\\
\hline
\hline
... & $\Box$ & $\underline{0}$ & $1$ & $1$ & $\Box$ & ...\\
\hline
\hline
... & $\Box$ & $\underline{0}$ & $1$ & $1$ & $\Box$ & ...\\
\hline
\hline
... & $\Box$ & $\underline{0}$ & $1$ & $1$ & $\Box$ & ...\\
\hline
\hline
... & $\Box$ & $\underline{0}$ & $1$ & $1$ & $\Box$ & ...\\
\hline
\hline
... & $\Box$ & $\underline{0}$ & $1$ & $1$ & $\Box$ & ...\\
\hline
\end{tabular}~~&~~
\begin{tabular}{c|c|c|c|c|c|c}
\multicolumn{7}{c}{~}\\
\hline
... & $\Box$ & $\underline{v}_1$ & $v_4$ & $\Box$ & \multicolumn{2}{c}{...}\\
\hline
\hline
... & $\Box$ & $\underline{1}$ & $0$ & $0$ & $\Box$ & ...\\
\hline
\hline
... & $\Box$ & $\underline{0}$ & $1$ & $1$ & $\Box$ & ...\\
\hline
\hline
... & $\Box$ & $\underline{0}$ & $1$ & $1$ & $\Box$ & ...\\
\hline
\hline
... & $\Box$ & $\underline{1}$ & $0$ & $0$ & $\Box$ & ...\\
\hline
\hline
... & $\Box$ & $\underline{0}$ & $1$ & $1$ & $\Box$ & ...\\
\hline
\end{tabular}~~&~~
\begin{tabular}{c|c|c|c|c|c|c}
\multicolumn{7}{c}{~}\\
\hline
... & $\Box$ & $v_1$ & $v_4$ & $\underline{\Box}$ & \multicolumn{2}{c}{...}\\
\hline
\hline
... & $\Box$ & $\underline{1}$ & $0$ & $0$ & $\Box$ & ...\\
\hline
\hline
... & $\Box$ & $0$ & $\underline{1}$ & $1$ & $\Box$ & ...\\
\hline
\hline
... & $\Box$ & $0$ & $\underline{1}$ & $1$ & $\Box$ & ...\\
\hline
\hline
... & $\Box$ & $\underline{1}$ & $0$ & $0$ & $\Box$ & ...\\
\hline
\hline
... & $\Box$ & $0$ & $1$ & $\underline{1}$ & $\Box$ & ...\\
\hline
\end{tabular}
\\
(b) & (c) & (d)\\

\end{tabular}
\end{center}
\caption{Computation of $(\{0\},\mathbb{N}^+)$\textsc{-Dominating Set Of Size At Most} $k$ on a blind multitape Turing machine whith $k=2$ on $C_5$(see proof of Theorem \ref{sigmarho:w2}). (a) Input graph; (b) State of the machine at the end of phase (1). The candidate set $D$ is on the first tape, the other tapes are initialized according to $\rho$; (c)  End of phase (2): the tapes associated with vertices in $D$ are now initialized according to $\sigma$; (d) End of phase (3): all heads (underlined symbols) read $1$, so $\{v_1,v_4\}$ is a $(\{0\},\mathbb{N}^+)$-dominating set.}
\label{figexample}
\end{figure*}

\begin{theorem}
For any recursive sets of integers $\sigma$ and $\rho$, $(\sigma,\rho)$-\textsc{Dominating Set of Size at Least} $n-k$ and $(\sigma,\rho)$-\textsc{Dominating Set of Size} $n-k$ belong to $\text{\emph{W}}[2]$.
\end{theorem}

\begin{proof} We prove that $(\sigma,\rho)$-\textsc{Dominating Set of Size at Least} $n-k$ is in W$[2]$, the proof that $(\sigma,\rho)$-\textsc{Dominating Set of Size} $n-k$ belongs to W$[2]$ is similar. To decide whether a given graph $G$ has a $(\sigma,\rho)$-dominating set of size at least $n\mathord -k$, we slightly modify the blind Turing machine used in the proof of Theorem \ref{sigmarho:w2} in such a way that at the end of phase (2), the first tape contains the description of a set $D$ of size at most $k$, and for any $v\in D$ (resp. $v\notin D$),  the $i^{th}$ cell of the tape associated with $v$ is $1$ if $\delta(v) \mathord - i \in \rho$ (resp. $\delta(v) \mathord - i \in \sigma)$ and $0$ otherwise, where $\delta(v)$ is the degree of $v$. 
Therefore, the machine reaches the accepting state if there exists a set $D$ of size at most $k$ such that  $\forall v\in D$, $\delta(v)-|N(v)\cap D|\in \rho$ and $\forall v\in V\setminus  D$, $\delta(v)-|N(v)\cap D|\in \sigma$. Since for any $v\in V$, $|N(v)\cap (V\setminus D)| = \delta(v)-|N(v)\cap D|$, $V\setminus D$ is a $(\sigma,\rho)$-dominating set of size at least $n\mathord-k$.\hfill $\Box$
\end{proof}

For any recursive sets $\sigma$ and $\rho$, $(\sigma,\rho)$-domination problems are in W$[2]$, but for some particular instances of $\sigma$ and $\rho$ this general result can be refined. In particular, we show that when $\sigma=\{0\}$ and $\rho=\{0,1\}$, the problem is W$[1]$-complete: 

\vspace{0.2cm}

\noindent\textsc{Strong Stable Set ($(\{0\},\{0,1\})$-Domination)}:\\
Input: A graph $G=(V,E)$, an integer $k$.\\
Parameter: $k$.\\
Question: Is there an independent set $S\subseteq V$ of size $k$ such that $\forall v\in V\setminus S, |N(v)\cap S|\le 1$? 

\begin{theorem}
\textsc{Strong Stable Set} is complete for $\text{\emph{W}}[1]$.
\end{theorem}

\begin{proof}
The $\text{W}[1]$-membership is an application of Theorem 8 in \cite{sigmarho}. We prove the hardness by a reduction from \textsc{Independent Set} which is complete for $\text{W}[1]$ \cite{fptc1}. Given an instance $(G=(V,E),k)$ of \textsc{Independent Set}, we consider the instance $(G',k)$ of \textsc{Strong Stable Set} where $G'=(V',E')$ with  $V'=V\cup E$ and $E'=\{(u,e)~|~\text{$e$ incident to $u$ in $G$}\}\cup (E\times E)$.
By construction, $G'$ consists of a stable set $V$ and a clique $E$, the edges between these two sets representing the edges of $G$. 
Let $S$ be an independent set in $G$, then by construction, $S$ is a strong stable set in $G'$. 
Let $S'$ be a strong stable set of size $k$ in $G'$. Since $E$ is a clique, $|S'\cap E|\in \{0,1\}$. If $|S'\cap E|=0$, then $S'\subseteq V$ and for any $u,v\in S'$, they have no common neighbor in $G'$, so there is no edge between $u$ and $v$ in $G$, so $S'$ is an independent set in $G$. Otherwise, if $|S'\cap E|=1$ then every $u\in S'\cap V$ is isolated in $G'$, so there are at least $k-1$ isolated vertices in $G$. Since $E$ is not empty there also exist non isolated vertices and we can take at least one of them to form together with the $k-1$ isolated vertices, an independent set of size $k$ in $G$. \hfill $\Box$
\end{proof}

\section{Other Domination Problems}
\label{otherdom}

Some natural domination problems cannot be described  in terms of $(\sigma,\rho)$-domina\-tion such as \textsc{Connected Dominating Set}. In this section, we show that the proof of the $(\sigma,\rho)$-domination W$[2]$-membership (Theorem \ref{sigmarho:w2}) can be generalized to $(P,\rho)$-domination, where $P$ is no longer a domination constraint but any recursive property. It implies that \textsc{Connected Dominating Set}, known to be hard for W$[2]$, is actually complete for W$[2]$. We also show that this technique can be applied to digraph problems with the example of \textsc{Digraph Kernel}.\vspace{0.2cm}\\
\noindent$(P,\rho)$-\textsc{Dominating Set of Size at Most} $k$:\\
Input: A graph $G=(V, E)$, an integer $k$.\\
Parameter: $k$.\\
Question: Is there a subset $D\subseteq V$ such that $|D|\leq k$ and:\\
\hspace{2cm}-- the sub-graph of $G$ induced by $D$ satisfies the property $P$;\\
\hspace{2cm}-- $\forall v\in V\setminus D, |N(v)\cap D|\in \rho$ ?

\begin{theorem}
If $\rho$ is a recursive set of integers and $P$ is a recursive property, then $(P,\rho)$-\textsc{Dominating Set of Size at Most} $k$ belongs to $\text{\emph W}[2]$.
\end{theorem}

\begin{proof}
We use the blind multitape Turing machine of Theorem \ref{sigmarho:w2} with $\sigma=\mathbb N$, which outputs a  $(\mathbb N, \rho)$-dominating set $D$ if it exists, then we compose this machine with another one which decides whether such a set $D$ induces a subgraph satisfying the property $P$. Since the subgraph is of size $O(k^2)$ and $P$ is recursive, the computation time of the second machine is $f(k)$ for some function $f$.   \hfill $\Box$
\end{proof}

\noindent\textsc{Digraph Kernel}:\\
Input: A directed graph $G=(V,A)$, an integer $k$.\\
Parameter: $k$.\\
Question: Is there a kernel of $D$ of size at most $k$? A \emph{kernel} is an independent set $S$ (there exists no $u,v\in S$ such that $uv$ or $vu$ is in $A$) such that for every vertex $x\in V\setminus S$, there exists $y\in S$ such that $xy\in A$.

\begin{theorem}
\textsc{Digraph Kernel} is complete for $\text{\emph W}[2]$.
\end{theorem} 

\begin{proof}
The hardness for $\text{W}[2]$ is proved in \cite{kernelDigraph}. The proof of the membership is very similar to the $\text{W}[2]$ membership of $(\sigma,\rho)$\textsc{-Dominating Set}  (Theorem \ref{sigmarho:w2}).
The machine and the initialization are the same, with $\sigma=\{0\}$ and $\rho=\mathbb{N}^+$. In phase (3), only the heads of the tapes associated with \emph{incoming neighbors} move to the right. \hfill $\Box$
\end{proof}

\section{Problems From Coding Theory}
\label{codes}

Parameterized complexity of problems from coding theory, in particular \textsc{Minimal Distance} and \textsc{Weight Distribution}, have been studied in \cite{paraCompCode}. We prove that the dual parameterizations of these problems are in W$[2]$. Moreover, we consider extensions of these problems to linear codes over $\mathbb F_{q}$ for any $q$ power of prime.

\vspace{0.2cm}

\noindent\textsc{Minimal Distance Over} $\mathbb{F}_{q}$:\\
Input: $q$ a power of prime, $k$ an integer, an $m \times n$ matrix $H$ with entries in $\mathbb{F}_{q}$.\\
Parameters: $k,q$.\\
Question: Is there a linear combination of at least one and at most $k$ columns of $H$ which is equal to the all-zero vector?

\vspace{0.2cm}

\noindent\textsc{Weight Distribution Over} $\mathbb{F}_{q}$:\\
Input: $q$ a power of prime, $k$ an integer, an $m \times n$ matrix $H$ with entries in $\mathbb{F}_{q}$.\\
Parameters: $k,q$.\\
Question: Is there a linear combination of exactly $k$ columns of $H$ which is equal to the all-zero vector?
\begin{theorem}\label{paraCompCode:w2}
\textsc{Weight Distribution Over} $\mathbb{F}_{q}$ is hard for $\text{\emph{W}}[1]$ and belongs to $\text{\emph{W}}[2]$, and \noindent\textsc{Minimal Distance Over} $\mathbb{F}_{q}$ belongs to $\text{\emph{W}}[2]$.
\end{theorem}

\begin{proof}
Since \textsc{Weight Distribution} is a particular case of \textsc{Weight Distribution Over} $\mathbb{F}_{q}$, with $q=2$, \textsc{Weight Distribution Over} $\mathbb{F}_{q}$ is hard for W$[1]$ \cite{paraCompCode}. 
For the $\text{W}[2]$ membership, let $\psi:[0,q)\to \mathbb F_q$ be an arbitrary indexing of the elements of $\mathbb F_q$ s.t. $\psi(0)=0$. There exist a prime $p$ and an integer $c$ such that $q=p^c$, and there is an isomorphism $\varphi : \mathbb F_q \to \mathbb F_p[X]/P(X)$, where $F_p[X]/P(X)$ is the set of polynomials in $X$ with coefficients in $\mathbb F_p$ modulo $P(X)$. 
Let $H'$ be a $mc\times (n(q{-}1))$-matrix over $\mathbb F_p$ such that $\forall i,j,\ell \in[0,m)\times [0,n)\times[1,q)$, $\sum_{u=0}^{c-1}  H'_{it,j\ell}X^t  = \varphi\left(  \psi(\ell)\cdot H_{i,j} \right)$. Intuitively, each of the $n(q-1)$ columns of $H'$ corresponds to one column of $H$ multiplied by a non-zero element of $\mathbb F_q$. Moreover any element $a\in \mathbb F_q$ is encoded using a $c\times 1$-block $\left(\begin{array}{c}r_0\\\vdots\\ r_{c-1}\end{array}\right)$  such that $\varphi(a) = \sum_{t=0}^{c-1} r_tX^t$. It leads to the  $mc\times (n(q{-}1))$-matrix $H'$ which can be computed in time $m.n.f(q)$ for some function $f$.

Notice that there exists a linear combination of $k$ columns of $H$ which is equal to $0$ if and only if there exist $0\le i_1<i_2<\ldots<i_k<m(q-1)$ such that the corresponding columns of $H'$ sums to $0$ (i.e. $\forall j\in [0,mc), \sum_{r=1}^k H'_{j,i_r} = 0$) and $\forall r\in [1,k), \left\lfloor \frac {i_r} m\right\rfloor \neq \left\lfloor \frac {i_{r+1}} m\right\rfloor$. The last condition guarantees that the $k$ chosen columns in $H'$ correspond to actually $k$ distinct columns in $H$. 

To decide whether such $i_1,\ldots,i_k$ exists we use the following  blind $(mc+1)$-tape Turing Machine $M=(Q,\Gamma,\Delta,\Sigma,b,q_I,Q_A)$. The first tape is associated with the set of columns of $H'$ and each of the remaining tape is associated with a row of $H'$. 
The alphabet is $\Gamma=\{\Box,0,1\}\cup\{h_i|i\in[1,n]\}$ and the states are $Q=\{q_{i,s}~|~i\in[1,n(q{-}1){+}1],s\in[0,k\cdot p]\} \cup \{q^{\textup{ret}}_{s}~|~s\in[1,k\cdot p+1]\} \cup \{q^{\textup{av}}_{i,s}~|~i\in[1,n],s\in[0,p-1]\} \cup \{q^{\textup{read}},q_A\}$, with $q_I=q_{1,0}$, $b=\Box$, $\Sigma=\varnothing$ and $Q_A=\{q_A\}$.
The transitions are separated in two phases:

\vspace{0.2cm}

\noindent \textit{Phase 1 - Initialization}: 
First, $k$ columns of $H'$ are non-deterministically chosen on the first tape, while of the other tapes is initialized with $k$ times the pattern $10^{p-1}$ (i.e. $1$ followed by $p-1$ times $0$), such that the $i^{th}$ cell is $1$ iff $i\equiv 0\mod p$. In order to avoid choosing two columns of $H'$ corresponding to the same column of $H$ but with a different factor, we go strait to the next block of columns, i.e. when a column $j$ is chosen, the next column is chosen in among the columns indexed from $\ell$ to $n(q-1)$ with $\ell>j$ and $\ell \equiv 0\bmod (q-1)$:
\vspace{0.2cm}

\noindent $\langle \Box\Box^{m\cdot c},q_{i,s},h_j 1^{m\cdot c},q_{\ell,s+1},(+1) (+1)^{m\cdot c} \rangle$\\
$~\quad i{\in}[1,n(q-1)], s{\in}[0,k-1], j{\in}[i,n(q-1)]$, if $s{\equiv} 0 (mod\ p)$\\
$~\quad \ell$ is the smallest integer such that $\ell>j$ and $\ell{\equiv} 0 (mod\ q-1)$\\~\\
\noindent$\langle \Box\Box^{m\cdot c},q_{i,s},h_j 0^{m\cdot c},q_{\ell,s+1},(+1) (+1)^{m\cdot c} \rangle$\\
$~\quad i{\in}[1,n(q-1)], s{\in}[0,k-1], j{\in}[i,n]$, if $s{\not\equiv} 0 (mod\ p)$\\
$~\quad \ell$ is the smallest integer such that $\ell>j$ and $\ell{\equiv} 0 (mod\ q-1)$\\~\\
$\langle \Box\Box^{m\cdot c},q_{i,s},\Box 1^{m\cdot c},q_{i,s+1},0 (+1)^{m\cdot c} \rangle$\\ $~\quad i {\in}[1,n(q{-}1){+}1], s{\in}[k,kp)$, if $s{\equiv} 0 (mod\ p)$\\
\\
$\langle \Box\Box^{m\cdot c},q_{i,s},\Box 0^{m\cdot c},q_{i,s+1},0 (+1)^{m\cdot c} \rangle$\\ $~\quad i{\in}[1,n{+}1], s{\in}[k,kp)$, if $s{\not\equiv} 0 (mod\ p)$\\
\\
\noindent\begin{tabular}{ll}
$\langle \Box\Box^{m\cdot c},q_{i,k\cdot p},\Box 1^{m\cdot c},q^{\textup{ret}}_{1},(-1) (-1)^{m\cdot c}\rangle$&$~i{\in}[1,n(q-1){+}1]$\\
&\\
$\langle \tspace\ \tspace^{m\cdot c},q^{\textup{ret}}_{s},\tspace\ \tspace^{m\cdot c},q^{\textup{ret}}_{s+1},(-1) (-1)^{m\cdot c}\rangle$&$~s{\in}[1,k]$\\
&\\
$\langle \tspace\ \tspace^{m\cdot c},q^{\textup{ret}}_{s},\tspace\ \tspace^{m\cdot c},q^{\textup{ret}}_{s+1},0 (-1)^{m\cdot c}\rangle$&$~s{\in}[k{+}1,kp{+}1]$\\
&\\
$\langle \tspace\ \tspace^{m\cdot c},q^{\textup{ret}}_{k\cdot p+1},\tspace\ \tspace^{m\cdot c},q^{\textup{read}},0 0^{m\cdot c}\rangle$\\
\end{tabular}

\vspace{0.2cm}

\noindent \textit{Phase 2 - Recognition}:  
In order to check that the sum of those columns is the all-zero vector on $\mathbb{F}_{p}$, for any column $h_i$ in the chosen set, the head of each tape $j$ moves to the right $H'_{i,j}$ times using blind transitions.

~\\ 

\begin{tabular}{ll}
$\langle h_i \tspace^{mc},q^{\textup{read}},h_i \tspace^{mc},q^{\textup{av}}_{i,1},(+1) 0^{mc}\rangle$&
$~ i\in[1,n]$\\
&\\
$\langle \tspace\ \tspace^{mc},q^{\textup{av}}_{i,s},\tspace\ \tspace^{mc},q^{\textup{av}}_{i,s+1},0 d_1 \dots d_{mc} \rangle$&
$~ i\in[1,n], s\in[0,p-2]$\\&with $\forall j{\in}[1,mc]$,$d_j=\begin{cases}1&\text{if $H'_{i,j}>s$}\\0&\text{otherwise}\end{cases}$\\
&\\
$\langle \tspace\ \tspace^m,q^{\textup{av}}_{i,p-1},\tspace\ \tspace^{mc},q^{\textup{read}},0 0^{mc} \rangle$&
$~ i\in[1,n]$\\
&\\
$\langle \Box 1^{mc},q^{\textup{read}},\Box 1^{mc},q_A,0 0^{mc}\rangle$&\\
\end{tabular}

~\\

In the first phase a set $D$ of columns is non-deterministically chosen on the first tape and on each  of the remaining tapes,  $kp$
 cells are filled with  $0$ or $1$ depending on the rest modulo $p$ of their position. Then all the heads move back to leftmost non blanc symbol. 
 Notice that the columns in $D$ are chosen to guarantee that $\forall i\neq i'\in D$, $\left\lfloor \frac {i} m\right\rfloor \neq \left\lfloor \frac {i'} m\right\rfloor$. In the second phase, the sum of the columns in $D$ is computed by moving the heads of the tapes to the right. The machine accepts iff at the end all the heads (but the first one) point out a symbol $0$, i.e. the sum of all the columns in $D$ of $H'$ is the zero vector. Regarding the number of transitions, in the first phase there are $2kp$ transitions and at most $kp$ in the second phase. Moreover the size of the machine is polynomial in $n$,$m$,$q$ and $k$. As a consequence \textsc{Weight Distribution} is in W$[2]$.  

The proof of W$[2]$-membership for \textsc{Minimal Distance} is the similar, except that $D$ is chosen of size at most $k$. \hfill $\Box$
\end{proof}

\noindent\textsc{Dual Minimal Distance Over} $\mathbb{F}_{q}$:\\
Input: $q$ a power of prime, $k$ an integer, an $m \times n$ matrix $H$ with entries in $\mathbb{F}_{q}$.\\
Parameters: $k,q$.\\
Question: Is there a linear combination of at least  $n-k$ columns of $H$ equal to the all-zero vector?

\vspace{0.2cm}

\noindent\textsc{Dual Weight Distribution Over} $\mathbb{F}_{q}$:\\
Input: $q$ a power of prime, $k$ an integer, an $m \times n$ matrix $H$ with entries in $\mathbb{F}_{q}$.\\
Parameters: $k,q$.\\
Question: Is there a linear combination of exactly $n-k$ columns of $H$ equal to the all-zero vector?

\begin{theorem}
\label{dualW2}
\textsc{Dual Minimum Distance Over $\mathbb{F}_{q}$} and \textsc{Dual Weight Distribution Over $\mathbb{F}_{q}$} are in $\text{\emph W}[2]$.
\end{theorem}

\begin{proof}
First we execute the same $FPT$ preprocessing as in standard parameterization (Theorem \ref{paraCompCode:w2}) to get the matrix $H'$ over $\mathbb F_p$ where $p$ is the characteristic of $\mathbb F_q$. Let the vector $v$ be the sum of all the columns of $H'$, and  notice that there is set $D$ of at $n-k$ columns that sum to the zero vector iff the sum of all the columns but those in $D$ sum to $v$. To this end we consider the matrix $\tilde H = \left(-v|H'\right)$ and we slightly modify the machine used in Theorem \ref{paraCompCode:w2} to decide whether the exists a set of at most $k+1$ columns which includes the first column and which sum to $0$. So the phase $1$ is modified to force the set of chosen columns to include the first column of $\tilde H$. 
The proof that \textsc{Dual Weight Distribution Over $\mathbb{F}_{q}$} is in W$[2]$ is the same except that the size of $S'$ is fixed to $k$.\hfill $\Box$
\end{proof}

Theorem \ref{dualW2} shows  that the problem \textsc{Minimal Distance Over} $\mathbb{F}_q$ with $q$ power of prime, which consists in deciding whether there exists a subset of at most $k$ columns of a matrix $H$ with entries in $\mathbb{F}_q$ that sum to the all-zero vector is in W$[2]$. We can prove similarly the  W$[2]$-membership of the problem which consists in deciding whether there exists a set of at most $k$ columns that sum to a given vector. However it is not clear whether the problem which consists in deciding the existence of a set of at most $k$ columns that sum to a vector with no zero entry (or equivalently to a vector of maximal Hamming weight). To be more precise when $q$ is prime one can use the same machine as in proof of Theorem   \ref{paraCompCode:w2} and change the last transition to check that non of the entries is $0$, but this technique fails when $q$ is not prime (say $q=p^2$).

\section{Conclusion and Perspectives}
We have demonstrated several results on the parameterized complexity of domina\-tion-type problems, including that for any (recursive)  $\sigma$ and $\rho$, $(\sigma,\rho)$-domination is in W$[2]$ for both standard and dual parameterizations i.e. $(\sigma,\rho)$-dominating set of size $k$ (and at most $k$) and $(\sigma,\rho)$-dominating set of size $n-k$ (and at least $n-k$). To this end, we have extended the Turing way to parameterized complexity with a new way to prove W$[2]$-membership using `blind' Turing machines. We believe that this machine  can be used to prove W$[2]$-membership of other problems, not necessarily related to domination.    

\noindent Several questions remain open. First, the long-standing question regarding the W$[1]$-hardness of \textsc{Minimal Distance} remains open \cite{paraCompCode}. Moreover, several problems related to domination with parity constraints, such as \textsc{Weight Distribution}, are W$[1]$-hard and in W$[2]$, are they complete for one of these two classes, or intermediate? This question is particularly interesting since these problems have been proved to form an equivalence class with other problems from quantum computing \cite{GJMP-MEMICS12,CP-WOD12}.

It is interesting to notice  that, for the dual parameterization, the difference between \textsc{Minimal Distance} and \textsc{Weight Distribution} seems to vanish in the sense that both problems are W$[1]$-hard, while the completeness for W$[1]$ or W$[2]$ remains open. In fact, no problem of $(\sigma,\rho)$-domination is known to be W$[2]$-complete for the dual parameterization, thus one can wonder if such a problem exists or if for any $\sigma$ and $\rho$, $(\sigma,\rho)$-domination is in W$[1]$ for the dual parameterization? It would be interesting to examine $\sigma$ and $\rho$ not ultimately periodic since they are among the  few known cases of hardness when the problem is parameterized by the tree-width \cite{sigRhoTreeWidth}.

\bigskip
\subsubsection*{Acknowledgments. } 
We would like to thank Sylvain Gravier and Mehdi Mhalla for several helpful discussions, and anonymous referees for fruitful comments. This work has been partially funded by the ANR-10-JCJC-0208 CausaQ grant and by the  Rh\^one-Alpes region.



\end{document}